\documentclass[11pt]{llncs}

\usepackage{amsmath, amssymb}

\usepackage{graphicx}
\usepackage{epsfig}
\usepackage{algpseudocode}

\usepackage{framed}
\usepackage{tabularx}

\newcommand{\problemConstruction}[3]
{
\begin{figure*}\noindent
 \begin{framed}
   \centerline{{\sc #1}}     
   \bigskip
   \noindent
   \begin{tabularx}{\textwidth}{lX}
     \emph{Given:} & #2	\\
     \emph{Task:} & #3
   \end{tabularx}
 \end{framed}
\end{figure*}
}

\newtheorem{observation}{Observation}

\title{Minimum Power Range Assignment for Symmetric Connectivity in Sensor Networks with two Power Levels}
\author{Stefan Hoffmann \and Egon Wanke}
\institute{\{Stefan.Hoffmann, E.Wanke\}@uni-duesseldorf.de \\Institute of Computer Science, Heinrich-Heine-Universit\"at D\"usseldorf\\ D-40225 D\"usseldorf, Germany}
\begin{document}
\maketitle
\begin{abstract}
This paper examines the problem of assigning a transmission power to every node of a wireless sensor network. The aim is to minimize the total power consumption while ensuring that the resulting communication graph is connected. We focus on a restricted version of this {\sc Range Assignment (RA)} problem in which there are two different power levels. We only consider symmetrical transmission links to allow easy integration with low level wireless protocols that typically require bidirectional communication between two neighboring nodes. We introduce a parameterized polynomial time approximation algorithm with a performance ratio arbitrarily close to $\pi^2/6$. Additionally, we give an almost linear time approximation algorithm with a tight quality bound of $7/4$.

\end{abstract}

\section{Introduction}

A sensor network consists of a large number of small devices that are deployed across a geographic area to monitor certain aspects of the environment. These sensor nodes are able to communicate with each other through a wireless communication channel. As a result of their size the resources of the sensor nodes are strongly limited in terms of available energy. This leads to a very limited range of the radio transmitters.

Minimizing the power consumption to prolong the lifetime of the network has received considerable attention, because the replacement of batteries in a large scale network is practically difficult or even impossible.

A common approach to save energy is based on adjusting the power levels of the sensor nodes' radio transmitters while maintaining a connected network structure. This {\sc Range Assignment (RA)} problem is typically defined as computing a range assignment $f: P\to \mathbb{R}^+$ for a set of points $P\subset\mathbb{R}^n$ ($1\le n\le 3$), representing the nodes in the network, such that the total energy $\sum_{p\in P} c(f(p))$ is minimal ($c$ being a cost function according to a radio wave propagation model) under the constraint that the graph $(P,E)$, $E:=\{(p,q)\in P^2 ~|~ \Vert p-q\Vert_2 \le f(p)\}$ is strongly connected, where $\Vert p-q\Vert_2$ denotes the Euclidean distance between $p$ and $q$.
This is called the {\em asymmetric} {\sc RA} problem.

Since many low level protocols for wireless communication like IEEE 802.11 require bidirectional transmission links for hop-by-hop acknowledgements \cite{IEEE80211}, the constraint of the asymmetric {\sc RA} problem is often altered to only consider bidirectional edges $E':=\{\{p,q\} ~|~ (p,q)\in E \wedge (q,p)\in E\}$ for an undirected connected graph. This problem is called the {\em symmetric} {\sc RA} problem.

\subsection*{Contribution}
This paper examines the symmetric version of the {\sc RA} problem in networks with two power levels (denoted by {\sc 2LSRA}). The $2$-dimensional decision version of the {\sc 2LSRA} problem is NP-complete, because the proof of \cite{CK07} for the asymmetric case can easily be modified to the symmetric case. We introduce an approximation algorithm (denoted by {\sc Approx2LSRA$_k$}) for the {\sc 2LSRA} problem parameterized by an integer $k\ge 2$. We proof that the performance ratio of {\sc Approx2LSRA$_k$} is between
\[\frac{3k-2}{2k-2} \quad \text{and} \quad \frac{1}{k-1}+\sum_{i=1}^{k-1} \frac{1}{i^2}.\]
Algorithm {\sc Approx2LSRA$_k$} can easily be implemented such that its running time is in $O(|V|^k)$, where $V$ is the set of sensor nodes. This provides a parameterized polynomial time approximation algorithm with a performance ratio arbitrarily close to $\pi^2/6$. For the case $k=3$, for which {\sc Approx2LSRA$_k$} has a tight performance of $7/4$, we provide an implementation running in almost linear time.

The approximation algorithm {\sc Approx2LSRA$_k$} is based solely on the graph structure and the following assumption: If a node $u$ can receive messages from a node $v$ while $v$ transmits with power $r_{\min}$, $u$ can still receive messages from $v$ if $v$ increases its power level to $r_{\max}$. Note that the algorithm requires neither any embedding of the nodes into $\mathbb{R}^d$ for some dimension $d$ nor any kind of correlation between the transmission links, as for example, if $u$ can receive messages from $v$ then $v$ can receive messages from $u$.

\subsection*{Related Work}
Both versions of the {\sc RA} problem are typically considered with cost function $c(d) = \beta\cdot d^\alpha + \gamma$ for a {\em distance-power gradient} $\alpha\ge 1$, a {\em transmission quality parameter} $\beta$ and some $\gamma\ge 0$. 
Kirousis et al.~proved in \cite{KKKP00} that the decision version of the asymmetric RA problem is NP-hard for 3-dimensional point sets for all $\alpha\ge 1$, $\beta = 1$ and $\gamma=0$. They also presented a factor $2$ approximation algorithm based on a minimum spanning tree. The hardness result was extended by Clementi et al.~in \cite{CPS99} to 2-dimensional point sets. Clementi et al.~also showed in \cite{CPS99} the APX-hardness of the 3-dimensional case. This implies that there is no polynomial time approximation scheme (PTAS) for the 3-dimensional asymmetric RA problem, unless P=NP. 
The 1-dimensional case, on the other hand, is solvable in time $\mathcal{O}(|P|^4)$ by dynamic programming techniques, see \cite{KKKP00}.

For the symmetric {\sc RA} problem, Blough et al.~showed in \cite{BLRS02} that the  decision version is also NP-hard for 2- and 3-dimensional point sets, whereas Althaus et al.~have shown in \cite{ACMPTZ03,ACMPTZ06} that there is a simple approximation algorithm with approximation factor $11/6$ as well as a parameterized polynomial time approximation algorithm with a performance ratio arbitrarily close to $5/3$.

Carmi and Katz introduced in \cite{CK07} the additional restriction of only two different transmission powers $r_{\min}$ and $r_{\max}$ with $r_{\min} < r_{\max}$. This simplifies the {\sc RA} problems to the question of how to determine a minimum set of sensor nodes that have to transmit with maximum power $r_{\max}$ in order to obtain connectivity. They study the $2$-dimensional asymmetric version of this problem, proof its NP-completeness, and present an approximation algorithm with a factor of $11/6$ as well as a rather theoretical approximation algorithm with a factor of $9/5$.

While formulated in an entirely different way, the {\sc 2LSRA} problem is basically equivalent to the problem called {\sc Max Power Users} given in \cite{LLR04}, which was renamed to {\sc $\{0,1\}$-MPST} by the authors of \cite{AZ09}. And although this work has been conducted independently of \cite{LLR04}, the concept used for an approximation algorithm is remarkably similar: In this paper a family of approximation algorithms, called {\sc Approx2LSRA$_k$} and based on a positive integer $k\in\mathbb{N}$, is developed and analyzed. The authors of \cite{LLR04} introduce two approximation algorithms which are in fact equivalent to our special cases {\sc Approx2LSRA$_3$} and {\sc Approx2LSRA$_4$}. However, the best known approximation ratio possibly for this problem is given in \cite{AZ09} where the authors prove that the algorithm presented in \cite{ACMPTZ06} for the more general problem of arbitrary power levels achieves an approximation ratio of $3/2$ for the special case of two power levels. It is noteworthy that the algorithm given in \cite{ACMPTZ06} is based on a rather complex approximation scheme for the classical {\sc Steiner Tree} problem, while the idea used in this paper and in \cite{LLR04} is a simple, fast and easy to implement greedy approach.

\section{Definitions and Terminology}

The problem we consider is motivated by wireless sensor networks in which the nodes have two transmission power levels, a {\em min-power} and a {\em max-power level}. Such networks are usually represented as directed graphs in that the directed edges represent connections from source nodes to target nodes. Since almost all communication protocols are based on symmetric connections, we are mainly interested in the underlying symmetric networks. These networks can be represented by undirected graphs.

Our goal is to find a minimum number $k$ of nodes such that if these $k$ nodes use max-power and the remaining nodes use min-power then the resulting underlying symmetric network is connected.

The problem is now defined more formally. We consider {\em undirected graphs} $G=(V,E)$, where $V$ is the set of {\em nodes} and $E \subseteq \{\{u,v\}~|~u,v \in V,~u \not=v\}$ is a set of {\em undirected edges}. A graph $G'=(V',E')$ is a {\em subgraph} of $G=(V,E)$ if $V' \subseteq V$ and $E' \subseteq E$. It is an {\em induced subgraph} of $G$, denoted by $G|_{V'}$, if $E' = E \cap \{\{u,v\}~|~ u,v \in V'\}$.

A {\em path} between two nodes $u,v \in V$ is a sequence of nodes $u_1,\ldots,u_k \in V$, $k \geq 1$, such that $u_1 = u$, $\{u_i,u_{i+1}\} \in E$ for $i=1,\ldots,k-1$ and $u_{k}=v$. Graph $G=(V,E)$ is {\em connected} if there is a path between every pair of nodes. A path $u_1,\ldots,u_k$, $k \ge 3$, is a {\em cycle} if $\{u_k,u_1\} \in E$, a graph without cycles is a {\em forest}, and a connected forest is a {\em tree}. A {\em connected component} of $G$ is a maximal induced subgraph of $G$ that is connected. Let $CC(G)$ be the set of all node sets of the connected components of $G$.

For a node $u \in V$, let $d_{\min} (u) \subseteq V$ and $d_{\max} (u) \subseteq V$ be the set of nodes reachable from $u$ with min-power or max-power, respectively. That is, $d_{\min}$ and $d_{\max}$ can be considered as mappings from $V$ to the power set ${\cal P}(V)$ of $V$, i.e., to the set of all subsets of $V$. 

\begin{definition}
\label{D01}
For a node set $V$, a subset $U \subseteq V$, and two mappings $d_{\min}:V \to {\cal P}(V)$ and $d_{\max}:V \to {\cal P}(V)$
define the {\em min-power edges between nodes in $U$} as 
\[E_{\min}(U) := \{\{u,v\}~|~u,v \in U, u\not=v, u \in d_{\min}(v), v \in d_{\min}(u)\} \text{,}\]
the {\em max-power edges between nodes in $U$} as 
\[E_{\max}(U) := \{\{u,v\}~|~u,v \in U, u\not=v, u \in d_{\max}(v), v \in d_{\max}(u)\} \text{ and }\]
the {\em min-max-power graph}
$G(V,d_{\min},U,d_{\max}):=(V,E_{\min}(V) \cup E_{\max}(U))$.
\end{definition}

The min-max-power graph $G(V,d_{\min},U,d_{\max})$ is also denoted by $G(U)$ if $V$, $d_{\min}$, and $d_{\max}$ are known from the context. We also assume that $d_{\min}(v) \subseteq d_{\max}(v)$ for all $v\in V$.

The min-max-power graph $G(U)$ represents the underlying symmetric network for the case that the nodes of $U$ use max-power and the remaining nodes of $V\setminus U$ use min-power. Graph $G(\emptyset)$ is also called the {\em min-power graph} whereas graph $G(V)$ is also called the {\em max-power graph}.

The definition of the min-max-power graph is used to define the {\sc 2-Level Symmetric Range Assignment} problem as follows.

\problemConstruction{2-Level Symmetric Range Assignment (2LSRA)}{
A node set $V$ and two mappings $d_{\min}: V\to{\cal P}(V)$ and $d_{\max}: V\to{\cal P}(V)$ such that the max-power graph $G(V)$ is connected.}{Compute a set $U\subseteq V$ of minimum cardinality such that the min-max-power graph $G(U)$ is connected.}

The decision problem of finding a node set $U$ of size at most $k$, for an additionally given integer $k$, such that $G(U)$ is connected is NP-complete. This result is implicitly contained in Theorem 6.1 of \cite{CK07}, which proofs the NP-completeness of the corresponding asymmetric version.


\section{An approximation algorithm for {\sc 2LSRA}}

If the min-max-power graph $G(U)$ is connected for some set $U \subseteq V$ then $U$ has to contain at least one node of every connected component of the min-power graph $G(\emptyset)$. That is, the number $|CC(G(\emptyset))|$ of connected components of $G(\emptyset)$ is a lower bound for the size of such a set $U$. On the other hand, it is easy to find a set $U \subseteq V$ with at most $2(|CC(G(\emptyset))|-1)$ nodes such that $G(U)$ is connected. Such a set can be determined by a simple spanning tree algorithm. Let $H$ be the graph that has a node for every connected component of $G(\emptyset)$ and an edge between two nodes $C,C' \in CC(G(\emptyset))$ if there are nodes $u \in C, u' \in C'$ such that $u' \in d_{\max}(u)$ and $u \in d_{\max}(u')$. Let $T$ be a spanning tree for $H$. Then for every edge $\{C,C'\}$ of $T$ we can select two nodes $u \in C, v \in C'$ from different connected components of $G(\emptyset)$ such that for the set $U$ of all these selected nodes the min-max-power graph $G(U)$ is 
connected.

\begin{observation}
\label{O01}
Any solution for an instance of {\sc 2LSRA} contains at least \newline$|CC(G(\emptyset))|$ and at most $2(|CC(G(\emptyset))| - 1)$ nodes.
\end{observation}

Our algorithm starts with the min-power graph $(V,E_0) := G(\emptyset)$ and an empty node set $U_0$ that is successively extended to node sets $U_i$  by adding node sets $M_i$ such that $(V,E_i)$ has less connected components than $(V,E_{i-1})$, where $E_i := E_{i-1} \cup E_{\max}(M_i)$. This is done until $(V,E_i)$ is connected, which implies that $G(U_i)$ is connected, because $(V,E_i)$ is a subgraph of $G(U_i)$. Obviously, a good choice for $M_i$ is a set for which the ratio
\[\frac{|CC((V,E_{i-1}))|-|CC((V,E_i))|}{|M_i|}\]
is high. For example, if each $M_i$ consists of two nodes then the number of connected components will be reduced by one at each extension and the algorithm computes a solution of size at most $2(|CC(G(\emptyset))| - 1)$. This is equivalent to a spanning tree solution. If each $M_i$ consists of three nodes and the number of connected components is reduced by two at each extension, then the algorithm computes a solution of size at most $\frac{3}{2}(|CC(G(\emptyset))| - 1)$.

\begin{definition}
A set of nodes $M \subseteq V$ of cardinality $k$ is called a {\em $k$-merging} for a graph $G=(V,E)$ and a mapping $d_{\max}:V \to {\cal P}(V)$, if
\begin{enumerate}
\item the graph $(M,E_{\max}(M))$ is connected, and
\item the $k$ nodes of $M$ are in $k$ different connected components of $G$.
\end{enumerate}
\end{definition}

If we add a $k$-merging $M_i$ to node set $U_{i-1}$ then $U_{i} = U_{i-1} \cup M_i$ and
\[\frac{|CC((V,E_{i-1}))|-|CC((V,E_i))|}{|M_i|} = \frac{k-1}{k}.\]

The approximation algorithm {\sc Approx2LSRA$_k$} shown in Figure \ref{approxK} has a fixed parameter $k \geq 2$. Starting with $k'=k$, it successively gathers $k'$-mergings as long as possible. After this it decreases $k'$ and proceeds in the same way until $(V,E_i)$ is connected, which will occur at the latest during the iteration for $k'=2$, where all remaining $2$-mergings are considered.

\begin{observation}
Algorithm {\sc Approx2LSRA$_k$} always finds a solution for an instance $V,d_{\min},d_{\max}$ of {\sc 2LSRA}.
\end{observation}

For every integer $k\in\mathbb{N}$ the algorithm {\sc Approx2LSRA$_k$} listed in Figure \ref{approxK} can be implemented such that its running time is in ${\mathcal O}(|V|^k)$, because the running time is dominated by the computation of all subsets $M\subseteq V$ with $|M|=k$ in line \ref{alg:lineChooseMerging}. Therefore we get a polynomial time algorithm for every constant integer $k$.

\begin{figure}[hbt]
\begin{algorithmic}[1]
\Function{Approx2LSRA$_k$}{$V,d_{\min},d_{\max}$}
\State $k'\gets k$
\State $i\gets 0$
\State $U_0 \gets \emptyset$
\State $E_0 \gets E_{\min}(V)$	
\While{$(V,E_i)$ is not connected}
  \While{there is a  $k'$-merging $M_i \subseteq V$ for $(V,E_i)$}	\label{alg:lineChooseMerging}
    \State $U_{i+1}\gets U_{i} \cup M_{i}$
    \State $E_{i+1}\gets E_{i} \cup E_{\max}(M_{i})$	\label{lineUpdateE}
    \State $i\gets i+1$
  \EndWhile
  \State $k'\gets k'-1$ \label{lineWhile}
\EndWhile
\State \Return $U_i$
\EndFunction
\end{algorithmic}\caption{Algorithm {\sc Approx2LSRA$_k$} for a fixed integer $k \geq 2$}\label{approxK}
\end{figure}

\section{Upper Bounds on the quality of {\sc Approx2LSRA$_k$}}\label{sectionAnalysis3}

Let $U_{OPT}(I)$ be an optimal solution for an instance $I=(V,d_{\min},d_{\max})$ of {\sc 2LSRA}. We show that {\sc Approx2LSRA$_k$} for a positive integer $k\ge 2$ computes a solution $U_k(I)$ such that
 \[\frac{|U_k(I)|}{|U_{OPT}(I)|} \quad \le \quad \frac{1}{k-1} + \sum_{i=1}^{k-1} \frac{1}{i^2}.\]

\begin{lemma}\label{L01}
Let $F_0=(V,E_0)$ be a forest with $n=|V|$ nodes and $m=|E_0|$ edges and let $p\in\mathbb{N}$, $1 \leq p \leq n-1$ and $l\in\mathbb{N}$ be positive integers.
\item Furthermore let $F_i = (V,E_i)$, $1\le i\le l$, be a sequence of forests such that $F_l$ contains only trees with less than $p$ edges and $E_i \subseteq E_{i-1}$, $|E_i| = |E_{i-1}|-p$.
\newline
If $m > \frac{n \cdot (p-1)}{p}$ then
\medskip
\begin{enumerate}
\item
  $F_0$ contains a tree with at least $p$ edges,
\item
each forest $F_i$,
\[i \quad < \quad \frac{1}{p} \cdot \left( m - \frac{n \cdot (p-1)}{p}\right),\]
contains a tree with at least $p$ edges and
\item
\[l \quad \geq \quad \left\lceil \frac{1}{p} \cdot \left( m - \frac{n \cdot (p-1)}{p}\right) \right\rceil.\]
\end{enumerate}
\end{lemma}

\begin{proof}
~
\begin{enumerate}
\item
Forest $F_0$ consists of $n-m$ trees. If $m > (n-m) \cdot (p-1)$, then at least one of these $n-m$ trees has more than $p-1$ edges and thus at least $p$ edges.
\[m > (n-m) \cdot (p-1) \quad \Leftrightarrow \quad m > \frac{n \cdot (p-1)}{p}\]

\item
Forest $F_i$ has $m-i \cdot p$ edges. If \[i < \frac{1}{p} \cdot \left( m - \frac{n \cdot (p-1)}{p}\right)\] then $F_i$ has more than $m - \frac{1}{p} \cdot \left( m - \frac{n \cdot (p-1)}{p}\right) \cdot p = \frac{n \cdot (p-1)}{p}$ edges and by Lemma \ref{L01}.1. at least one tree with $p$ edges.
\item
Follows from 2. and the fact that we still can remove at least $p$ more edges if there is a tree with $p$ edges.
\end{enumerate}\qed
\end{proof}

\begin{lemma}\label{lemma:mergings-size3}
Let $k\in\mathbb{N}, k\ge3$ be a positive integer and $U_{\text{OPT}}$ an optimal solution for an instance of {\sc 2LSRA}. Then {\sc Approx2LSRA$_k$} always finds at least 
\[\left\lceil \frac{1}{k-1} \cdot \left( (|CC(G(\emptyset))|-1) - \frac{|U_{\text{OPT}}| \cdot (k-2)}{k-1}\right) \right\rceil\]
$k$-mergings.
\end{lemma}

\begin{proof}
Let $M_0,\ldots,M_{l-1}$ be the $k$-mergings chosen by {\sc Approx2LSRA$_k$} in Line \ref{alg:lineChooseMerging} for a given instance $V,d_{\min},d_{\max}$ of {\sc 2LSRA} and $E_1,\ldots, E_{l}$ the edge sets computed in line \ref{lineUpdateE} starting with edge set $E_0$ of the min-power graph $G(\emptyset)$.

Let $H=(CC(G(\emptyset)),E_H)$ be the undirected graph that has a node for every connected component of $G(\emptyset)$ and an edge $\{C_1,C_2\}$ if and only if there is an edge $\{u,v\}$ in the min-max-power graph $G(U_{\text{OPT}})$ with $u \in C_1$ and $v \in C_2$.

For $i=0, \ldots, l$ we define a tree $T_i$ and a forest $F_i=(U_{\text{OPT}},E'_i)$ that satisfies the following invariant:
\begin{quote}
(I1) The node set $R$ of every tree of forest $F_i$, $0 \leq i \leq l$, is a $|R|$-merging for graph $(V,E_i)$.
\end{quote}

Let $T_0=(CC(G(\emptyset)),E_T)$ be a spanning tree of $H$ and $F_0$ be a forest that contains for every edge $\{C_1,C_2\} \in E_T$ exactly one edge $\{u,v\}$ of $G(U_{\text{OPT}})$ with $u \in C_1$ and $v \in C_2$. Invariant (I1) above obviously holds true for $F_0$.

For every $k$-merging $M_i=\{u_1,\ldots,u_k\}$, $0 \leq i < l$, we successively define trees $T_{i,1},\ldots,T_{i,k}$ and forests $F_{i,1},\ldots,F_{i,k}$ starting with $T_{i,1} := T_{i}$ and $F_{i,1} := F_{i}$ such that $T_{i+1} := T_{i,k}$ and $F_{i+1} := F_{i,k}$.

For $j = 2,\ldots,k$, tree $T_{i,j}$ and forest $F_{i,j}$ are defined by merging two nodes of tree $T_{i,j-1}$ and removing one edge from forest $F_{i,j-1}$, respectively. 

Let $C^*$ be the node of $T_{i,j-1}$ that contains the nodes $u_{1},\ldots,u_{j-1}$ and let $C$ be the node of $T_{i,j-1}$ that contains $u_j$. 
Choose any edge $\{C',C''\}$ from the path between $C^*$ and $C$ in $T_{i,j-1}$, replace the two nodes $C^*,C$ by one new node $C^* \cup C$ in $T_{i,j-1}$, and remove edge $\{u,v\}$ with $u \in C'$ and $v \in C''$ from $F_{i,j-1}$.

If Invariant (I1) holds true for $F_i$, then it holds true for $F_{i+1}$, because the construction above guarantees that for every simple path $v_1, \ldots ,v_m$ in $F_{i+1}$ the nodes $v_i$, $1\le i\le m$, are in $m$ different connected components of $(V,E_{i+1})$. 

Since forest $F_0$ has $|U_{\text{OPT}}|$ nodes and $|CC(G(\emptyset))|-1$ edges, by Lemma \ref{L01}.3, algorithm {\sc Approx2LSRA$_k$} finds at least 
\[\left\lceil \frac{1}{k-1} \cdot \left( (|CC(G(\emptyset))|-1) - \frac{|U_{\text{OPT}}| \cdot (k-2)}{k-1}\right) \right\rceil\]
$k$ mergings.
\qed
\end{proof}

\begin{theorem}\label{thmUpper}
 Let $I = (V,d_{\min},d_{\max})$ be an instance of {\sc 2LSRA} and $U_{OPT}(I)$ an optimal solution. For a fixed integer $k\ge 2$, algorithm {\sc Approx2LSRA$_k$} computes a solution $U_k(I)$ such that
 \[\frac{|U_k(I)|}{|U_{OPT}(I)|} \quad \le \quad \frac{1}{k-1} + \sum_{i=1}^{k-1} \frac{1}{i^2}.\]
\end{theorem}
\begin{proof}
 Let $l_i$, $3\le i\le k$, denote the number of $i$-mergings chosen by {\sc Approx2LSRA$_k$}. Furthermore, let $c_k := |CC(G(\emptyset))|$ be the number of connected components of the min-power graph and let $c_i$, $2 \le i \le k-1$, be the number of connected components before {\sc Approx2LSRA$_k$} searches for $i$-mergings for the first time, that is $c_i := |CC((V,E_{s(i)}))|$ with $s(i):=\sum_{j=i+1}^{k}l_j = l_k+\ldots + l_{i+1}$
 . Then we know for $3\le i\le k$ that
 \begin{equation}
    c_{i-1} = c_{i} - (i-1)l_{i} \quad \Leftrightarrow \qquad l_i = \frac{c_i - c_{i-1}}{i-1}.\label{equationCL}
 \end{equation}
 Let $d_{\min,s(i)}: V\to{\cal P}(V)$ be defined by $u\in d_{\min,s(i)}(v)$ if and only if $\{u,v\}\in E_{s(i)}$ for all $u,v\in V$. Then $(V,E_{s(i)})$ is the min-power graph of instance $I_{s(i)} := (V,d_{\min,s(i)},d_{\max})$. Now we can apply Lemma \ref{lemma:mergings-size3} for $k=i$ on instance  $I_{s(i)}$ and get
 \begin{align}
     \qquad \frac{c_i - c_{i-1}}{i-1}	\quad &\ge \quad \frac{1}{i-1} \left(c_i - 1 - \frac{i-2}{i-1}\cdot |U_{OPT}(I_{s(i)})|\right)		\notag\\
    \Leftrightarrow  \qquad c_{i-1} - 1 \quad &\le \quad \frac{i-2}{i-1}\cdot |U_{OPT}(I_{s(i)})|	 \quad \le \quad  \frac{i-2}{i-1} \cdot |U_{OPT}(I)|	,	\label{equationCLE}
 \end{align}
 because $|U_{OPT}(I_{s(i)})| \le |U_{OPT}(I)|$. We can derive an upper bound for $|U_k(I)|$ as follows.
 \begin{align*}
  |U_k(I)| &\quad\le\quad \sum_{i=3}^{k}i\cdot l_i + 2(c_2-1) \quad\stackrel{(\ref{equationCL})}{=}\quad \sum_{i=3}^{k}i\cdot l_i + 2\left(c_k-1-\sum_{i=3}^{k}(i-1)l_i\right) 	\\
	  &\quad=\quad 2(c_k-1) - \sum_{i=3}^{k}(i-2)l_i \quad\stackrel{(\ref{equationCL})}{=}\quad 2(c_k-1) - \sum_{i=3}^{k}\frac{i-2}{i-1}\left(c_i - c_{i-1}\right) 	\\
	  &\quad=\quad 2(c_k-1) + \sum_{i=3}^{k}\frac{i-2}{i-1}\left(c_{i-1}-1\right) - \sum_{i=3}^{k}\frac{i-2}{i-1}\left(c_i-1\right)	\\
	  &\quad=\quad 2(c_k-1) + \frac{c_2-1}{2} - \frac{(k-2)(c_k-1)}{k-1} + \sum_{i=3}^{k-1}\left(\frac{i-1}{i}-\frac{i-2}{i-1}\right)\left(c_{i}-1\right) \\
	  &\quad=\quad \left(2-\frac{k-2}{k-1}\right)\cdot(c_k-1) + \frac{c_2-1}{2} + \sum_{i=3}^{k-1}\frac{c_{i}-1}{i(i-1)} \\
	  &\quad\stackrel{(\ref{equationCLE})}{\le}\quad \left(2-\frac{k-2}{k-1}\right)\cdot(c_k-1) + \frac{|U_{OPT}(I)|}{4} + \sum_{i=3}^{k-1}\frac{|U_{OPT}(I)|}{i^2}
 \end{align*}
Since $c_k-1 < |U_{OPT}(I)|$, see Observation \ref{O01}, we get
 \begin{align*}
  |U_k(I)| &< \left(2-\frac{k-2}{k-1} + \frac{1}{4} + \sum_{i=3}^{k-1}\frac{1}{i^2}\right)	\cdot|U_{OPT}(I)| = \left(\frac{1}{k-1}+ \sum_{i=1}^{k-1}\frac{1}{i^2}\right)	\cdot|U_{OPT}(I)|.
 \end{align*}\qed
\end{proof}
\begin{corollary}
{\sc Approx2LSRA$_3$} is a $7/4$ factor approximation algorithm for {\sc 2LSRA}.
\end{corollary}
\begin{observation}
  The upper bound in Theorem \ref{thmUpper} tends to $\pi^2/6$ for $k\to\infty$. \cite{A83}
\end{observation}

\section{Lower Bounds on the quality of {\sc Approx2LSRA$_k$}}
In this section we present worst case examples for algorithm {\sc Approx2LSRA$_k$} and derive lower bounds on its quality. For $k=3$ we demonstrate that the upper bound of $7/4$ obtained in section \ref{sectionAnalysis3} is tight.
\begin{theorem}\label{theorem-worst-case}
 Let $k\in\mathbb{N}$, $k\ge 3$ be a positive integer. For an instance $I$ of {\sc 2LSRA} let $U_k(I)$ be the solution computed by {\sc Approx2LSRA$_k$} and $U_{OPT}(I)$ an optimal solution.
 For all $\epsilon > 0$ there is an instance $I$ such that
 \[\frac{|U_k(I)|}{|U_{OPT}(I)|} > \frac{3k-2}{2k-2} - \epsilon \text{.}\]
\end{theorem}
\begin{proof}
  \begin{figure}[b]
  \null\hfill
  \includegraphics[width=\textwidth]{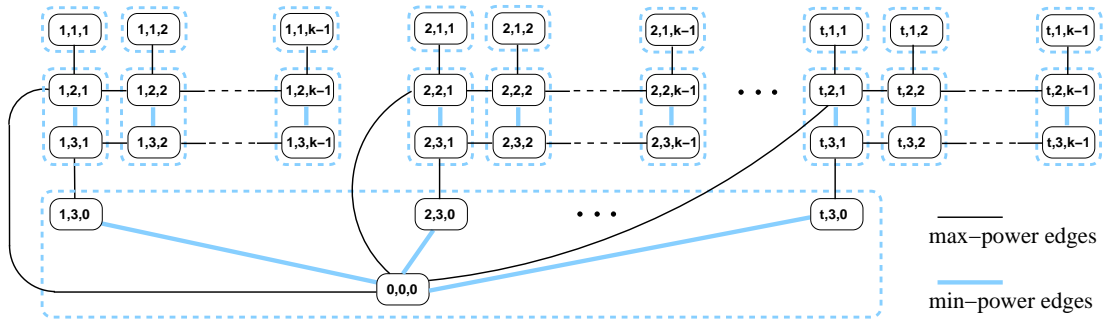}
  \hfill\null
  \caption{Instance $I_t$ of Theorem \ref{theorem-worst-case}: The min-power edges $E_{\min}(V)$ are drawn as thick blue lines, the max-power edges $E_{\max}(V)$ are drawn as thin black lines and the connected components of the min-power graph $G(\emptyset)$ are encircled by dashed boxes.}
  \label{approx-worst-case}
  \end{figure}
 For a positive integer $t\in\mathbb{N}$ let $I_t$ be the instance for {\sc 2LSRA} defined as follows (see Figure \ref{approx-worst-case}).
 \begin{align*}
  V := & \{(0,0,0)\} \quad \cup \quad \{(d,3,0) ~|~ 1\le d\le t\} 	\\
		& \cup \quad \{(d,r,c) ~|~ 1\le d\le t, 1\le r\le 3, 1\le c < k\}	\\
  E_{\min}(V) := & \{\{(0,0,0),(d,3,0)\} ~|~ 1\le d\le t\} 	\\
		 & \cup \{\{(d,2,c),(d,3,c)\} ~|~ 1\le d\le t, 1\le c<k\}	\\
  E_{\max}(V) := & E_{\min}(V) \quad \cup \quad \{\{(0,0,0),(d,2,1)\} ~|~ 1\le d\le t\} \\
		  & \cup \quad \{\{(d,r,c),(d,r,c+1)\} ~|~ 1\le d\le t, r\in\{2,3\}, 1\le c<k-1\}	\\
		  & \cup \quad \{\{(d,3,0),(d,3,1)\} ~|~ 1\le d\le t\}	\\
		  & \cup \quad \{\{(d,1,c),(d,2,c)\} ~|~ 1\le d\le t, 1\le c< k\}
 \end{align*}
  The (unique) optimal solution for $I_t$ is 
  \[U_{OPT}(I_t) = \{(0,0,0)\} \cup \{(d,r,c) ~|~ 1\le d\le t, r\in\{1,2\}, 1\le c<k\},\]
  meaning that $|U_{OPT}(I_t)| = 1 + 2(k-1)t$. Algorithm {\sc Approx2LSRA$_k$}, in the worst case, successively gathers all $k$-mergings $\{(d,3,c) ~|~ 0\le c < k\}$ for $1\le d\le t$ first, followed by the remaining $2$-mergings $\{(d,1,c),(d,2,c)\}$ for $1\le d\le t$ and $1\le c<k$. Thus we get
   \[q(k,t) := \frac{|U_k(I_t)|}{|U_{OPT}(I_t)|} = \frac{kt + 2(k-1)t}{1 + 2(k-1)t} = \frac{k + 2(k-1)}{\frac{1}{t} + 2(k-1)} 
   = \frac{3k-2}{\frac{1}{t} + 2k-2}\]
  and $\lim_{t \to \infty}(q(k,t)) = \frac{3k-2}{2k-2}$ 
  , which implies the existence of a positive integer $t_{\epsilon}\in\mathbb{N}$ such that $q(k,t_{\epsilon}) > \frac{3k-2}{2k-2} - \epsilon$.
\qed
\end{proof}
\begin{corollary}
 The upper bound of $7/4$ on the quality of {\sc Approx2LSRA$_3$} is tight.
\end{corollary}
\section{Efficient Implementation of {\sc Approx2LSRA$_3$}}

For $k=3$ the algorithm can be implemented in almost linear time using a union-find data structure $D$ to organize the node sets of the connected components of $(V,E_i)$. 

We assume that the mappings $d_{\min}$ and $d_{\max}$ are explicitly given as relations $d_{\min} \subset V\times {\cal P}(V)$ and $d_{\max} \subset V\times {\cal P}(V)$ of size $s_{\min} := |V| + \sum_{v\in V} |d_{\min}(v)|$ and $s_{\max} := |V| + \sum_{v\in V} |d_{\max}(v)|$, respectively. The implementation is based on the following three steps.
\begin{enumerate}
 \item {\em Initialization}: Generate the set of min-power edges $E_{\min}(V)$ and the set of max-power edges $E_{\max}(V)$. This can be done in time ${\mathcal O}(s_{\min})$ and ${\mathcal O}(s_{\max})$, respectively. Afterwards $D$ can be initialized with the connected components of $G(\emptyset)$ during one iteration over the min-power edges by performing two find operations $u'=find(u)$, $v'=find(v)$ and one union operation $union(u',v')$ for each $\{u,v\}\in E_{\min}(V)$.
 
 The total number of $find$ and $union$ operations for this step is linear in $|V|$ and $|E_{\min}(V)|$ and therefore also linear in $s_{\min}$.
 \item {\em Finding $3$-mergings:} For every node $v\in V$ examine all incident edges $\{v,u\}\in E_{\max}(V)$ after identifying the connected component of $v$ via $C_v = find(v)$. If $C_u := find(u) \not= C_v$, node $u$ and its component $C_u$ are temporarily saved until a second node $u'$ adjacent to $v$ is found, such that $C_v \not= C_{u'}$ and $C_u \not= C_{u'}$. Then $\{v,u,u'\}$ is a $3$-merging that is added to the solution and the three connected components $C_v$, $C_u$ and $C_{u'}$ are merged via two union operations. 
 
 Since every edge in $E_{\max}(V)$ has to be considered only twice, the total number of $find$ and $union$ operations for this step is linear in $|V|$ and $|E_{\max}(V)|$ and therefore also linear in $s_{\max}$.
 \item {\em Finding $2$-mergings:} For every edge $\{u,v\}\in E_{\max}(V)$ add $\{u,v\}$ to the solution and call $union(u,v)$, if $find(u) \not= find(v)$. Again, the total number of $find$ and $union$ operations for this step is linear in $s_{\max}$.
\end{enumerate}
\begin{theorem}
 {\sc Approx2LSRA$_3$} can be implemented such that the running time is in $\mathcal{O}(f(s_{\min},s_{\max})\cdot \alpha(f(s_{\min},s_{\max}),|V|))$ where $\alpha$ is the inverse Ackermann function and $f\in {\mathcal O}(s_{\min}+s_{\max})$.
\end{theorem}
\begin{proof}
 The total number of $union$ and $find$ operations $f(s_{\min},s_{\max})$ is linear in $s_{\min}$ and $s_{\max}$ as discussed above. Tarjan and Leeuwen show in \cite{TL84} that any sequence of $m\in\mathbb{N}$ $union$ and $find$ operations on a union-find data structure saving $n\in\mathbb{N}$ elements can be performed in time $\mathcal{O}(m\cdot\alpha(m,n))$. Fredman and Saks show in \cite{FS89} that this bound in tight.\qed
\end{proof}

\clearpage
\bibliographystyle{plain}
\bibliography{literature}
\end{document}